\documentclass[12pt,english]{article}
\usepackage{a4}
\usepackage[round]{natbib}
\usepackage{babel}

\usepackage[margin=1.25in]{geometry}
\usepackage{array}
\usepackage{amsmath,amssymb,bm}
\usepackage{amsthm}
\usepackage{mathtools}
\usepackage{amsbsy}
\usepackage{extarrows}
\usepackage[colorlinks = true,linkcolor = purple,urlcolor  = blue,citecolor = blue,anchorcolor = blue]{hyperref}
\usepackage{setspace}
\usepackage{xfrac}
\usepackage{tikz}
\usetikzlibrary{positioning}
\usepackage[affil-it]{authblk}

\usepackage{mathptmx} 

\begin{document}

\setlength{\extrarowheight}{4pt} 

\title{Nested Removal of Strictly Dominated Strategies in Infinite Games\footnote{Many thanks to Julien Manili for stimulating correspondence and insightful comments. I am also grateful to the Associate Editor and anonymous reviewers for helpful comments. Financial support from the OP Pohjola Research Foundation (Grant Nos. 20230084 and 20240089) is gratefully acknowledged. All errors are my own.}}
\author{Michele Crescenzi\\ \href{mailto:michele.crescenzi@helsinki.fi}{michele.crescenzi@helsinki.fi}}
\affil{University of Helsinki and Helsinki Graduate School of Economics}

\date{March 2026\\
\vspace*{10pt}
\textcolor{red}{This version is published in \textit{Mathematical Social Sciences} (2026), DOI: \url{https://doi.org/10.1016/j.mathsocsci.2026.102552}}}

\newtheorem{proposition}{Proposition}
\newtheorem{theorem}{Theorem}
\newtheorem*{theorem*}{Theorem}
\newtheorem*{definition*}{Definition}
\newtheorem{assumption}{Assumption}
\newtheorem{definition}{Definition}
\newtheorem{lemma}{Lemma}
\newtheorem{claim}{Claim}
\newtheorem{remark}{Remark}
\newtheorem*{remark*}{Remark}
\newtheorem{corollary}{Corollary}
\theoremstyle{definition}
\newtheorem{example}{Example}
\newtheorem*{example*}{Example}

\maketitle

\begin{abstract}
We compare two procedures for the iterated removal of strictly dominated strategies. In the nested procedure, a strategy of a player is removed only if it is dominated by an unremoved strategy. The universal procedure is more comprehensive for it allows the removal of strategies that are dominated by previously removed ones. Outside the class of finite games, the two procedures may lead to different outcomes in that the universal one is always order independent while the other is not. We provide necessary and sufficient conditions for the equivalence of the two procedures. The conditions we give are based on completely bounded subsets of strategy profiles, which are variations of the bounded mechanisms from the literature on full implementation. The two elimination procedures are also shown to be equivalent in quasisupermodular games as well as in games with compact strategy spaces and upper semicontinuous payoff functions. Finally, we give an example of a game in which the nested procedure is order independent but not equivalent to the universal one, thus showing that order independence is not sufficient to guarantee the equivalence of the two elimination procedures.

\bigskip

JEL CLASSIFICATION: C70, C72
\bigskip

\end{abstract}

\newpage

\onehalfspacing

\section{Introduction}
There are two procedures for the iterated removal of strictly dominated strategies in normal-form games. In the \textit{nested} procedure, a strategy of a player is removed only if it is dominated by an unremoved strategy. In any round of elimination, one can ignore the strategies that were previously removed. Thus, the removal of a strategy can be interpreted as its ``physical'' removal from the game. The \textit{universal} procedure is more comprehensive as it allows the elimination of any strategy that is dominated by an unremoved or even a previously removed strategy. Hence, there might be elimination rounds in which it is necessary to recall a removed strategy to eliminate an unremoved one. As a result, the removal of a strategy can no longer be interpreted as its ``physical'' removal from the game, thus adding to the computational burden of the universal procedure.

Which elimination procedure should be followed? In finite games the two procedures lead to the same outcomes, thus making it obvious to choose the nested one thanks to its simplicity. But the choice is not so obvious in infinite games, i.e., games with infinite strategy spaces, in which the two procedures may converge to different limits. As \cite{chen2007iterated-strict-dominance} show, the universal elimination procedure is order independent since in every game it converges to a unique product set of serially undominated strategies. What is more, the latter set is exactly the set of all strategy profiles compatible with common knowledge of not playing strictly dominated strategies. By contrast, \citet[Examples 1 and 5]{dufwenberg2002existence-maximal-reduction} produce infinite games in which the nested procedure\footnote{The elimination procedure in \cite{dufwenberg2002existence-maximal-reduction} is allowed to attain a limit in countably many steps, whereas in this paper we study elimination sequences of any transfinite length. This difference, though, does not affect order dependence in the two examples we are referring to.} fails on both counts, thus leaving some serially dominated strategies unremoved. Therefore, one should follow the universal procedure in any game where the nested one is flawed. In all other games the two procedures are equivalent and, just like in finite games, the nested one is preferable in that it is simpler and computationally less burdensome. But what are the games where the two elimination procedures are equivalent? Put differently, what are the games in which the nested procedure does not leave any serially dominated strategy unremoved? As the extant literature does not provide a full characterization of these games, we will provide one here.

Why look for a full characterization of the games in which the two elimination procedures are equivalent? For one thing, the characterization will help us gain a better understanding of the iterated removal of strictly dominated strategies, which is a fundamental technique for solving non-cooperative games. For another, we will be able to evaluate the necessity of different choices made in the literature and save on computational resources. For example, \cite{dufwenberg2002existence-maximal-reduction} and several game theory textbooks follow the nested elimination procedure. By contrast, in \cite{milgrom-roberts1990rationalizability} and the ensuing literature on supermodular games, the universal procedure is usually adopted. With the full characterization under our belt, we will be able to determine whether one can substitute the nested procedure for the universal one in the games just mentioned, thus economizing on computational resources.

We study the equivalence between the nested and universal removal of strictly dominated strategies for a vast class of normal-form games, in which the cardinalities of player sets and strategy spaces are unrestricted. We do not make any topological or measure-theoretic assumptions. Given such a high level of generality, we focus on strict domination by pure strategies. A procedure for the iterated removal of strictly dominated strategies is represented by non-increasing sequences of reductions. A reduction of a game is any product subset of strategy profiles of the game. Any round of elimination corresponds to a reduction, and in any round players can remove as many strictly dominated strategies as they wish. Elimination sequences can have transfinite length, and the limit of an elimination sequence is called its maximal reduction. The nested elimination procedure in this paper extends the elimination procedure of \cite{dufwenberg2002existence-maximal-reduction} to any transfinite length; the universal procedure we study is the IESDS* procedure of \cite{chen2007iterated-strict-dominance}, but with the minor  difference that we require game reductions to be in Cartesian product form.

To see how nested and universal elimination procedures may lead to different outcomes, consider a one-player game in which the strategy set is the open interval $(0,1)$ and the payoff function is the identity map.\footnote{This elementary game is also discussed in \citet[p. 2011]{dufwenberg2002existence-maximal-reduction} and \citet[p. 302]{chen2007iterated-strict-dominance}.} Every strategy is strictly dominated. The unique maximal reduction of the universal elimination procedure is the empty set. But since one does not have to remove all dominated strategies in one go, the nested procedure gives rise to infinitely many maximal reductions: the empty set and any singleton $\{a\}$, with $a \in (0,1)$, are all maximal nested reductions. For instance, $\{\sfrac{1}{2}\}$ is the maximal reduction of the nested elimination sequence in which $(\sfrac{1}{2}, 1)$ is removed in the first round and $(0, \sfrac{1}{2})$ in the second.

We show that the existence of completely bounded reductions is \textit{exactly} the condition that guarantees the equivalence between nested and universal elimination procedures. A reduction is completely bounded if, whenever a player has a strategy $a$ in the reduction that is strictly dominated by another strategy $b$, then there is an undominated strategy $c^*$ (possibly equal to $b$) that dominates $a$. Completely bounded reductions are variations of the bounded mechanisms introduced by \cite{jackson1992bounded} in the full implementation literature. Our paper is not the first to leverage the idea of boundedness. At least since \cite{dufwenberg2002existence-maximal-reduction}, several conditions have been found under which the nested elimination procedure does not lead to unsatisfactory outcomes. Such conditions are typically sufficient but not necessary.\footnote{The closest conditions to this paper are discussed in Section \ref{subsec:RelConditions}.} By contrast, we formalize boundedness so as to obtain conditions that are both necessary and sufficient for the equivalence between nested and universal procedures. Specifically, a crucial aspect in our definition of a completely bounded reduction is that the dominating strategies $b$ and $c^*$ are not required to belong to the reduction, i.e., they just need to be in the player's strategy space. As we show in Section \ref{sec:GKZ}, if one defines a bounded reduction requiring that $b$ and $c^*$ belong to the reduction itself, then nested and universal procedures are not necessarily equivalent. 

Our first result (Theorem \ref{thm:TFAE_outcomes}) is that the two elimination procedures lead to the same maximal reduction if and only if every maximal reduction of the nested procedure is completely bounded. In the one-player game discussed above, none of the non-empty maximal nested reductions $\{a\}$ is completely bounded, whence the discrepancy with the universal procedure. 

There are games in which the two procedures lead to the same maximal reduction, and yet the classes of nested and universal elimination sequences do not coincide. To wit, there exist universal elimination sequences that cannot be obtained in the nested procedure. A case in point is the game in Example \ref{ex:UnboundedAtLimit}. This calls for a stronger form of equivalence, in which every elimination sequence obtained with one procedure can also be obtained with the other. The stronger equivalence subsumes the one discussed in the previous paragraph. In Theorem \ref{thm:TFAE_classes} we establish that the classes of nested and universal elimination sequences coincide if and only if the range of every nested elimination sequence---or, equivalently, of every universal sequence---consists only of completely bounded reductions. An immediate corollary is that the strong equivalence holds in games where all reductions, including those not in the range of any elimination sequence, are completely bounded. We show in Proposition \ref{prop:AllBounded} that all reductions are completely bounded in: 1) games with finite strategy spaces; 2) games with compact strategy spaces and upper semicontinuous payoff functions; and 3) quasisupermodular games, which generalize the supermodular games of \cite{milgrom-roberts1990rationalizability}. In these three classes of games, we can use the nested elimination procedure without worrying.

Our findings help us understand better the role played by order independence, which is a desirable property of any iterated elimination procedure. We have already mentioned that the universal procedure is order independent while the nested one is not. As a consequence, the two procedures can be equivalent only in games where the nested procedure is order independent. But while it is necessary, order independence of the nested procedure is not sufficient for the equivalence. In Example \ref{ex:OInd-no-equal} we produce a game in which the nested procedure is order independent, and yet its unique maximal reduction differs from the one of the universal procedure, meaning that the maximal nested reduction contains serially dominated strategies. To the best of our knowledge, this is the first example in the literature of a game in which the nested and the universal elimination procedures attain different limits in spite of both being order independent. Conceptually, Example \ref{ex:OInd-no-equal} raises a note of caution: order independence of an elimination procedure does not guarantee satisfactory predictions.\footnote{\cite{manili2026} provides sufficient conditions that, for a wide range of solution concepts, guarantee both the order independence of an iterated elimination procedure and its epistemic validity, i.e., its consistency with epistemic foundations.}

We also compare the nested procedure with the elimination procedure introduced by \cite{gilboa1990order}. The latter is a nested procedure satisfying the additional requirement that, if a strategy $a$ is removed in a round of elimination, then there must be a strategy $b$ that strictly dominates $a$ and is not removed before the next round. We establish in Theorem \ref{thm:TFAE_GKZ} that the two elimination procedures lead to the same maximal reductions in every game, provided that elimination sequences of transfinite length are allowed. Furthermore, the class of elimination sequences \textit{\`{a} la} \cite{gilboa1990order} is contained in the class of nested sequences. The two classes coincide if and only if, in either type of procedure, the range of every elimination sequence consists only of reductions satisfying \textit{local} boundedness, which is a weaker version of the complete boundedness discussed above.

\paragraph*{Related work.} This paper contributes to the literature on the iterated removal of strictly dominated strategies in infinite games. \cite{dufwenberg2002existence-maximal-reduction} study the nested elimination procedure and give sufficient conditions for its order independence and the non-emptiness of its maximal reduction. Contrary to our work, they do not allow elimination sequences of transfinite length. The universal procedure is studied by \cite{chen2007iterated-strict-dominance}, who show its order independence in every game and give epistemic foundations for it. Our work combines the two contributions just mentioned by characterizing the equivalence between the nested and the universal elimination procedure. We also characterize the equivalence between the nested procedure and the elimination procedure of \cite{gilboa1990order}. An analogous characterization is provided by \cite{hsieh-et-al_iterated-bounded-dominance2023} for three classes of elimination sequences: nested, boundedly dominated, and elimination sequences \textit{\`{a} la} \cite{gilboa1990order}. We discuss their work in more detail in Section \ref{subsec:RelConditions}.

At least since \cite{gilboa1990order}, order independence has been a central theme in the literature on iterated elimination procedures. Our main contribution to this topic is to show that order independence is necessary but not sufficient for the equivalence between nested and universal elimination procedures. By contrast, the extant literature has focused on finding sufficient conditions for the order independent removal of strictly dominated strategies. Such conditions have been provided by \cite{dufwenberg2002existence-maximal-reduction}, \cite{gilboa1990order}, \cite{apt2007many-faces, apt2011direct}, \cite{luo2020iterated}, and \cite{patriche2013maximalreductiongames}. In Section \ref{sec:Discuss} we discuss in more detail how our work relates to the most relevant conditions found in the papers just mentioned. We do not delve into the epistemic foundations of order independence, which can be found in \cite{trost2014epistemic-rationale-or-ind}. Sufficient conditions for the order independence and the epistemic validity of several iterated elimination procedures are provided by \cite{manili2026}.

Iterated elimination procedures have been studied in a wide range of contexts. Our work focuses on the iterated removal of strictly dominated strategies in normal form games of complete information. Recently, iterated strict dominance has been studied in games of incomplete information by \cite{bach-perea2021incomplete}, whereas \cite{luo2020iterated} have examined iterated elimination procedures in abstract choice settings. \cite{maniliOrderIndependence} has explored the order independence of elimination procedures for rationalizability, whereas \cite{manili2026} has provided sufficient conditions that guarantee the equivalence of nested and universal elimination procedures for several solution concepts, including iterated strict dominance by mixed strategies and rationalizability.

\paragraph*{Paper structure.} We lay out the model in Section \ref{sec:Model}. In Section \ref{sec:Res} we study the equivalence between nested and universal elimination procedures and then provide examples to illustrate the main points of our analysis. In Section \ref{sec:GKZ} we compare our formulation of nested elimination procedures with the one of \cite{gilboa1990order}. In Section \ref{sec:Discuss} we discuss some epistemic aspects of the two main procedures studied in this paper as well as some conditions found in the literature that are close to our work.

\section{Model}\label{sec:Model}
\subsection{Game reductions and strict dominance}
Let $\Gamma = \left(I, (A_i)_{i\in I}, (u_i)_{i\in I}\right)$ be a game in normal form. The set of players is $I$. We assume there are at least two players\footnote{Assuming that there are at least two players is without loss of generality. Any one-player game can be turned into a two-player game by adding a dummy player who has a singleton set of strategies.} but the cardinality of $I$ is otherwise unrestricted. The set of pure strategies of player $i\in I$ is $A_i$, and $i$'s payoff function is the map $u_i$ from $A = \prod_{j\in I} A_j$ to $\mathbb{R}$. Strategy sets must be non-empty but their cardinalities are otherwise unrestricted. As usual, we write $u_i(a_i, a_{-i})$ to denote $u_i(a)$, where $a = (a_j)_{j\in I}$ is an element of $A$, $a_i$ is the $i$th term of $a$, and $a_{-i} = (a_j)_{j\in I\setminus\{i\}}$.

A \textbf{reduction} of $\Gamma$ is a subset $R \subseteq A$ such that $R = \prod_{i\in I} R_i$ and $R_i \subseteq A_i$ for all $i\in I$. For any $i\in I$ and any reduction $R$, we abbreviate $\prod_{j\in I\setminus\{i\}} R_j$ by $R_{-i}$. We often write $R$ as $R_i \times R_{-i}$. Any non-empty $R$ induces a reduced game of $\Gamma$, in which the set of players is still $I$, the strategy set of player $i\in I$ is $R_i$, and $i$'s payoff function is the restriction of $u_i$ to $R$.

Given a non-empty $R_{-i}$, a strategy $a_i \in A_i$ is \textbf{strictly dominated} by $b_i \in A_i$ relative to $R_{-i}$ if $u_i(a_i, a_{-i}) < u_i(b_i, a_{-i})$ for all $a_{-i} \in R_{-i}$; in this case we also write $a_i \prec_{R_{-i}} b_i$. Note that $\prec_{R_{-i}}$ defines a strict order---i.e., an asymmetric and transitive binary relation---over $A_i$. Our study is confined to strict dominance by a pure strategy. We often omit the qualifier \textit{strict} as no confusion should arise. The \textbf{dominating set} (or \textbf{strict upper contour set}) of $a_i$ relative to $R_{-i}$ is the set $D_{R_{-i}}(a_i):= \left\lbrace b_i\in A_i : a_i \prec_{R_{-i}} b_i \right\rbrace$. Transitivity of $\prec_{R_{-i}}$ implies $D_{R_{-i}}(b_i) \subseteq D_{R_{-i}}(a_i)$ for all $b_i \in D_{R_{-i}}(a_i)$. An \textbf{undominated} (or \textbf{maximal}) \textbf{element} of $D_{R_{-i}}(a_i)$ is a strategy $b_i^* \in D_{R_{-i}}(a_i)$ for which there is no $b_i \in D_{R_{-i}}(a_i)$ such that $b_i^* \prec_{R_{-i}} b_i$.

The iterated elimination procedures we are going to study are based on two binary relations defined on the set of all reductions of $\Gamma$. Formally, given two reductions $R$ and $S$, we write $R \xlongrightarrow{n} S$ and say that $S$ is a \textbf{nested reduction} of $R$, or $R$ reduces to the nested reduction $S$, if the following holds: $S \subseteq R$ and, for all $i\in I$,
\begin{equation}\label{eq:Nested-Red}
\text{if } a_i \in R_i\setminus S_i, \text{ then there is } b_i \in R_i \text{ such that } a_i \prec_{R_{-i}} b_i.
\end{equation}
Analogously, we write $R \xlongrightarrow{u} S$ and say that $S$ is a \textbf{universal reduction} of $R$, or $R$ reduces to the universal reduction $S$, if the following holds: $S \subseteq R$ and, for all $i\in I$,
\begin{equation}\label{eq:Universal-Red}
\text{if } a_i \in R_i\setminus S_i, \text{ then there is } b_i \in A_i \text{ such that } a_i \prec_{R_{-i}} b_i.
\end{equation}

Nested and universal reductions differ in how they constrain the dominating strategy $b_i$. In \eqref{eq:Nested-Red}, the dominating strategy $b_i$ must lie in $R_i$, which is a subset of the strategy set $A_i$. In \eqref{eq:Universal-Red}, $b_i$ is allowed to lie outside $R_i$. Clearly, a nested reduction is also universal but not vice versa.
Nested reductions are used by \cite{dufwenberg2002existence-maximal-reduction} and several textbooks, such as \citet[pp. 60-1]{ORgametheory94}, \citet[p. 45]{FTgametheory91} and \citet[pp. 58-9]{myerson91GTbook}. Universal reductions are used by \cite{chen2007iterated-strict-dominance} and \cite{kunimoto-serrano2011}.\footnote{ \cite{milgrom-roberts1990rationalizability}, \cite{moulin1979dominance} and \citet[pp. 193-6]{ritzberger2002GTbook} adopt a slightly less general version of universal reductions, in which each $R_i\setminus S_i$ must contain \textit{all} the strategies of player $i$ that are dominated relative to $R_{-i}$. By contrast, $R_i\setminus S_i$ in \eqref{eq:Nested-Red} and \eqref{eq:Universal-Red} can be any subset of those dominated strategies.}

In Section \ref{sec:GKZ} we examine an alternative formulation of nested reductions introduced by \cite{gilboa1990order}, in which the dominating strategy $b_i$ is required to be in $S_i$.

\subsection{Iterated elimination sequences and maximal reductions}\label{subsec:ElimSequen}
A procedure for the iterated removal of dominated strategies is represented by a class of elimination sequences. An elimination sequence is an ordinal-indexed sequence of reductions. We consider primarily two classes of sequences: nested and universal. Formally, an \textbf{elimination sequence} is a sequence $\langle R^{\alpha}: \alpha < \delta\rangle$, where $\delta$ is a non-zero ordinal,\footnote{The index ordinal $\delta$ can be either a successor or a limit ordinal. We can write $\langle R^{\alpha}: \alpha \leq \delta\rangle$ when the sequence is indexed by a successor ordinal $\delta + 1$.} such that:
	\begin{itemize}
	\item $R^0 = A$;
	\item $R^{\alpha} \xlongrightarrow{\ell} R^{\alpha +1}$ for all $\alpha$ such that $\alpha + 1 < \delta$;
	\item $R^{\lambda} = \cap_{\alpha < \lambda} R^{\alpha}$ for all limit ordinals $\lambda < \delta$;
	\item there exists a reduction $\hat{R}\subseteq A$ such that $R^{\alpha} = \hat{R}$ for some $\alpha < \delta$ and, for all reductions $S\subseteq A$, we have $\hat{R} \xlongrightarrow{\ell} S$ only if $S = \hat{R}$.
\end{itemize}

We do not put any restriction on the index ordinal $\delta$ other than it being different from zero.\footnote{Elimination sequences of transfinite length are employed, among others, by \cite{apt2007many-faces}, \cite{chen2007iterated-strict-dominance}, \cite{hsieh-et-al_iterated-bounded-dominance2023}, \cite{lipman1994note-transfinite}, and \cite{luo2020iterated}.} An elimination sequence is nested when $\ell = n$ in the definition above, whereas it is universal when $\ell = u$. Any elimination sequence is non-increasing, i.e., $R^{\beta} \subseteq R^{\alpha}$ whenever $\alpha < \beta$. The reduction $\hat{R}$ is the \textbf{maximal reduction} of the sequence. The classes $\mathcal{E}^n$ and $\mathcal{E}^u$ contain all feasible nested and universal elimination sequences, respectively. Both $\mathcal{E}^n$ and $\mathcal{E}^u$ are non-empty.\footnote{Non-emptiness of $\mathcal{E}^u$ follows from Theorem 1 in \citet[p. 304]{chen2007iterated-strict-dominance}, which holds true even when, as we assume here, every reduction must be in product form. In addition, the proof of their theorem can be easily adapted to show the non-emptiness of $\mathcal{E}^n$. Relatedly, the possible nonexistence of maximal reductions in \cite{dufwenberg2002existence-maximal-reduction} stems from the fact that only elimination sequences indexed by the first infinite ordinal $\omega$ are allowed in their study, whereas we allow sequences of any transfinite length.} The two classes $\mathcal{E}^n$ and $\mathcal{E}^u$ need not overlap; for instance, they are disjoint in Example \ref{ex:OInd-no-equal}. The set of all maximal reductions of all sequences in $\mathcal{E}^{\ell}$ is denoted by $\mathbf{\hat{R}}^{\ell}$. The sets $\mathbf{\hat{R}}^n$ and $\mathbf{\hat{R}}^u$ need not overlap either; for instance, they are disjoint in Example \ref{ex:OInd-no-equal}.

The \textbf{range} of an elimination sequence $\langle R^{\alpha}: \alpha < \delta\rangle$ is the set $\left\lbrace R^{\alpha}: \alpha < \delta \right\rbrace$. The range of a class $\mathcal{E}^{\ell}$ is the set of all reductions that belong to the range of some sequence in $\mathcal{E}^{\ell}$. An \textbf{initial segment} of $\langle R^{\alpha}: \alpha < \delta \rangle$ is a subsequence $\langle S^{\alpha}: \alpha < \delta' \rangle$, with $0< \delta' \leq \delta$, such that $R^{\alpha} = S^{\alpha}$ for all $\alpha < \delta'$. Any nested elimination sequence is the initial segment of some universal sequence.

A class of elimination sequences is \textbf{order independent} if all the sequences in the class have the same maximal reduction. \citet[Theorem 1]{chen2007iterated-strict-dominance} show that the class $\mathcal{E}^{u}$ is order independent in every game\footnote{More precisely, \citet[Theorem 1]{chen2007iterated-strict-dominance} show that the class of IESDS* elimination sequences, in which reductions are not necessarily in product form, is order independent. Since it is easy to prove that the class of universal sequences $\mathcal{E}^{u}$ is contained in the class of IESDS* sequences, the order independence of $\mathcal{E}^{u}$ is an immediate corollary of \citet[Theorem 1]{chen2007iterated-strict-dominance}.} and, as a consequence, $\mathbf{\hat{R}}^{u}$ is always a singleton. \cite{dufwenberg2002existence-maximal-reduction} produce infinite games in which the class of nested elimination sequences is instead order dependent. While $\mathbf{\hat{R}}^n$ and $\mathbf{\hat{R}}^u$ may well be disjoint, it is always the case that $\hat{R}^u \subseteq \hat{R}^n$ for all $\hat{R}^n \in \mathbf{\hat{R}}^n$, where $\hat{R}^u$ is the unique reduction in $\mathbf{\hat{R}}^u$. Epistemic foundations for the two elimination procedures are discussed in Section \ref{subsec:Epistem}.

\section{Results}\label{sec:Res}
Here we address the central question of the paper: When are the nested and the universal elimination procedures equivalent? There are two notions of equivalence. We can deem the two types of elimination procedures as equivalent if they both lead to the same set of maximal reductions, i.e., if $\mathbf{\hat{R}}^n = \mathbf{\hat{R}}^u$. A more restrictive notion of equivalence requires that the two classes of elimination sequences $\mathcal{E}^{n}$ and $\mathcal{E}^{u}$ be equal. Clearly, the second form of equivalence implies the first. We are going to show that both forms depend on the following condition, which adapts the notion of bounded mechanisms \citep{jackson1992bounded} for strict dominance.

\begin{definition}\label{Def:CBound}
A reduction $R$ is \textbf{completely bounded} if for all $i\in I$ and all $a_i \in R_i$ the following holds: if $D_{R_{-i}}(a_i)$ is non-empty, then it contains an undominated element. That is, if $a_i \prec_{R_{-i}} b_i$ for some $b_i \in A_i$, then there is a strategy $c_i^* \in A_i$ such that $a_i \prec_{R_{-i}} c_i^*$ and, for all $c_i \in A_i$, it is not the case that $c^*_i \prec_{R_{-i}} c_i$.

A class of elimination sequences $\mathcal{E}^{\ell}$ is completely bounded if every reduction in the range of $\mathcal{E}^{\ell}$ is completely bounded.
\end{definition}

The introduction of completely bounded reductions is motivated by the following observation. Consider again the one-player game discussed in the Introduction, in which the strategy set is the open interval $(0,1)$ and the payoff function is the identity map. Clearly, every strategy is dominated and the unique maximal universal reduction is empty. Nevertheless, the game has infinitely many maximal nested reductions. In addition to the empty set, any singleton $\{a\}$, with $a \in (0,1)$, is a maximal nested reduction too. A common feature of all these non-empty maximal reductions is that they are not completely bounded. This suggests that the presence of unbounded maximal reductions is the culprit for the discrepancy between universal and nested elimination procedures.

The maximal universal reduction is always (vacuously) completely bounded. If a maximal nested reduction $\hat{R}^n$ is completely bounded, then $\hat{R}^n = \hat{R}^u$. Thus $\mathbf{\hat{R}}^n$ can contain at most one completely bounded reduction, which coincides with the maximal universal reduction. The following lemma says that a maximal nested reduction can contain a dominated strategy only if it is not completely unbounded.

\begin{lemma}\label{lemma:MaxNestUnbound}
Let $\hat{R}$ be a maximal nested reduction. If for some player $i\in I$ there is a strategy $a_i \in \hat{R}_i$ such that $D_{\hat{R}_{-i}}(a_i)$ is not empty, then $D_{\hat{R}_{-i}}(a_i)$ does not contain any undominated element.
\end{lemma}
\begin{proof}
Take any nested elimination sequence having $\hat{R}$ as its maximal reduction. Suppose by way of contradiction that $D_{\hat{R}_{-i}}(a_i) \neq \emptyset$ and $b_i \in D_{\hat{R}_{-i}}(a_i)$ is undominated. Since $\hat{R}$ is maximal, the strategy $b_i$ must belong to $A_i \setminus \hat{R}_i$, meaning that $b_i$ must have been eliminated at some stage of the sequence at hand. Hence, there is a strategy $c_i \in A_i$ and a reduction $R$ in the range of our sequence such that $c_i \in R_i$ and $b_i \prec_{R_{-i}} c_i$. Since $\hat{R} \subseteq R$, we have $b_i \prec_{\hat{R}_{-i}} c_i$ and, by transitivity, $a_i \prec_{\hat{R}_{-i}} c_i$. This contradicts the assumption that $b_i$ is an undominated element of $D_{\hat{R}_{-i}}(a_i)$.
\end{proof}

Our first result is a full characterization of when the two sets of maximal reductions $\mathbf{\hat{R}}^n$ and $\mathbf{\hat{R}}^u$ coincide.

\begin{theorem}\label{thm:TFAE_outcomes}
The following statements are equivalent:
\begin{enumerate}
	\item[1)] $\mathbf{\hat{R}}^n = \mathbf{\hat{R}}^u$;
	\item[2)] Every maximal nested reduction $\hat{R}\in \mathbf{\hat{R}}^n$ is completely bounded;
	\item[3)] $\mathcal{E}^{n} \subseteq \mathcal{E}^{u}$.
\end{enumerate}
\end{theorem}
\begin{proof}
\begin{itemize}
\item $\bm{1) \implies 2).}$ Suppose $\mathbf{\hat{R}}^n = \mathbf{\hat{R}}^u$. By order independence of the class of universal elimination sequences, there is a unique maximal reduction $\hat{R}$, which is both universal and nested. By way of contradiction, suppose $\hat{R}$ is not completely bounded. This implies that for some player $i\in I$ and strategy $a_i \in \hat{R}_i$, the dominating set $D_{\hat{R}_{-i}}(a_i)$ is non-empty. Therefore, we have $\hat{R} \xlongrightarrow{u} \left(\hat{R}_i \setminus\{a_i\}\right) \times \hat{R}_{-i}$, which contradicts the maximality of $\hat{R}$.
\item $\bm{2) \implies 3).}$ Suppose $2)$ holds. Take any nested elimination sequence in $\mathcal{E}^{n}$ and let $\hat{R}^n$ be its maximal reduction. Recall that every nested sequence is the initial segment of some universal sequence. Now take any reduction $S$ such that $\hat{R}^n \xlongrightarrow{u}S$. Since $\hat{R}^n$ is completely bounded by assumption, and by Lemma \ref{lemma:MaxNestUnbound}, we must have $\hat{R}^n = S$, which shows that our nested sequence is universal too.
\item $\bm{3) \implies 1).}$ If every nested elimination sequence is also universal, then every maximal nested reduction must be a maximal universal reduction too. Then the equality $\mathbf{\hat{R}}^n = \mathbf{\hat{R}}^u$ follows easily from the order independence of the class of universal elimination sequences.
\end{itemize}
\end{proof}

Theorem \ref{thm:TFAE_outcomes} says that the two sets of maximal reductions $\mathbf{\hat{R}}^n$ and $\mathbf{\hat{R}}^u$ coincide if and only if complete boundedness holds ``in the limit,'' i.e., in every maximal nested reductions. Dually, the theorem says that $\mathbf{\hat{R}}^n \neq \mathbf{\hat{R}}^u$ if and only if there is a maximal nested reduction that is unbounded. The complete boundedness, or lack thereof, of non-maximal reductions is irrelevant. This point is illustrated by the game in Example \ref{ex:UnboundedAtLimit}, in which $\mathbf{\hat{R}}^n = \mathbf{\hat{R}}^u$ and yet some elimination sequences involve non-maximal reductions that are not completely bounded. Theorem \ref{thm:TFAE_outcomes} is slightly more general than Corollary 2 in \cite{chen2007iterated-strict-dominance}, which contains a sufficient condition for $\mathbf{\hat{R}}^n  =  \mathbf{\hat{R}}^u$ in terms of stable sets.

Since the class of universal elimination sequences is order independent, the equality $\mathbf{\hat{R}}^n = \mathbf{\hat{R}}^u$ holds true only if the class of nested sequences $\mathcal{E}^{n}$ is order independent. But we show in Example \ref{ex:OInd-no-equal} that the order independence of $\mathcal{E}^{n}$ is not sufficient for having $\mathbf{\hat{R}}^n = \mathbf{\hat{R}}^u$.

The next theorem characterizes the stronger form of equivalence between nested and universal elimination procedures.

\begin{theorem}\label{thm:TFAE_classes}
The following statements are equivalent:
\begin{enumerate}
\item[1)] The class of nested elimination sequences $\mathcal{E}^n$ is completely bounded;
\item[2)] The class of universal elimination sequences $\mathcal{E}^u$ is completely bounded;
\item[3)] $\mathcal{E}^n = \mathcal{E}^u$.
\end{enumerate}
\end{theorem}
\begin{proof}
\begin{itemize}
\item $\bm{1) \implies 3).}$ Suppose $1)$ holds. It follows from Theorem \ref{thm:TFAE_outcomes} that $\mathcal{E}^{n} \subseteq \mathcal{E}^{u}$. In order to show the reverse inclusion, take any universal elimination sequence $\langle R^{\alpha}: \alpha < \delta \rangle$. We are going to show by induction that every initial segment of the latter sequence is also the initial segment of a nested elimination sequence, from which we can conclude that $\langle R^{\alpha}: \alpha < \delta \rangle$ belongs to $\mathcal{E}^{n}$. The claim is clearly true for the initial segment of length $1$, whose only element is $R^0 = A$. Now let $0<\delta' < \delta$ and suppose $\langle R^{\alpha}: \alpha < \delta' \rangle$ is the initial segment of a nested elimination sequence. We need to prove that also $\langle R^{\alpha}: \alpha \leq \delta' \rangle$ is the initial segment of a nested sequence. Two cases are possible. In the first, $\delta'$ is a limit ordinal. Then by definition $R^{\delta'} = \cap_{\alpha < \delta'}R^{\alpha}$. Since $\langle R^{\alpha}: \alpha < \delta' \rangle$ is part of a nested sequence by assumption, it follows immediately that $\langle R^{\alpha}: \alpha \leq \delta' \rangle$ is the initial segment of a nested sequence. In the second case, $\delta'$ is a successor ordinal. Then there is a unique $\gamma < \delta'$ such that $R^{\delta'} = R^{\gamma + 1}$. In addition, we have $R^{\gamma} \xlongrightarrow{u} R^{\gamma + 1}$ by assumption. If $R^{\gamma} = R^{\gamma + 1}$, then it is immediate that $R^{\gamma} \xlongrightarrow{n} R^{\gamma + 1}$. If $R^{\gamma} \neq R^{\gamma + 1}$, then there are a player $i\in I$ and strategies $a_i\in R_i^{\gamma}\setminus R_i^{\gamma + 1}$ and $b_i \in A_i$ such that $a_i \prec_{R_{-i}^{\gamma}} b_i$. By the inductive hypothesis, $R^{\gamma}$ is in the range of $\mathcal{E}^{n}$ and so it is completely bounded by assumption. Hence, $D_{R_{-i}^{\gamma}}(a_i)$ must contain an undominated element, which cannot be removed at any stage $\alpha
\leq \gamma$ and must be in $R_i^{\gamma}$. Therefore, also in this case we have $R^{\gamma} \xlongrightarrow{n} R^{\gamma + 1}$. This shows that every initial segment of $\langle R^{\alpha}: \alpha < \delta \rangle$ is the initial segment of a nested sequence. Hence, $\langle R^{\alpha}: \alpha < \delta \rangle$ is in $\mathcal{E}^n$.
\item $\bm{3) \implies 2).}$ Suppose by way of contradiction that $\mathcal{E}^{n} = \mathcal{E}^{u}$ and a reduction in the range of $\mathcal{E}^{u}$ is not completely bounded. That is, there are a reduction $R$, a player $i\in I$, and a strategy $a_i \in R_i$ such that $D_{R_{-i}}(a_i)$ is not empty and does not contain any undominated element. Then we have
\begin{equation*}
R \xlongrightarrow{u} \left(R_i \setminus D_{R_{-i}}(a_i)\right) \times R_{-i} \xlongrightarrow{u} \left(R_i \setminus \left(D_{R_{-i}}(a_i)\cup \{a_i\}\right)\right) \times R_{-i},
\end{equation*}
but $\left(R_i \setminus \left(D_{R_{-i}}(a_i)\cup \{a_i\}\right)\right) \times R_{-i}$ is not a nested reduction of $\left(R_i \setminus D_{R_{-i}}(a_i)\right) \times R_{-i}$. This contradicts the assumption $\mathcal{E}^{n} = \mathcal{E}^{u}$.
\item $\bm{2) \implies 1).}$ Since every nested elimination sequence is the initial segment of a universal sequence, the range of $\mathcal{E}^n$ is a subset of the range of $\mathcal{E}^u$, from which the claim follows.
\end{itemize}
\end{proof}

Theorem \ref{thm:TFAE_classes} says that all feasible elimination sequences are both nested and universal if and only if the range of each feasible sequence consists only of completely bounded reductions. Dually, the existence of an unbounded reduction in the range of some elimination sequence is both necessary and sufficient for having $\mathcal{E}^n \neq \mathcal{E}^u$.

An immediate consequence of Theorem \ref{thm:TFAE_classes} is the following.

\begin{corollary}\label{cor:AllBound}
If all reductions of a game $\Gamma$ are completely bounded, then $\mathcal{E}^n = \mathcal{E}^u$.
\end{corollary}

The complete boundedness of all reductions of a game is a sufficient, but not necessary, condition\footnote{Example \ref{ex:NotAllBounded} shows that it is not a necessary condition.} for the strong equivalence between nested and universal elimination procedures. There are games in which it is easier to check the complete boundedness of all reductions---including those not belonging to the range of any elimination sequence---than it is to check the complete boundedness of all feasible elimination sequences. Before listing some of these games in the next proposition, we need a few definitions.

We say that $\Gamma$ is \textbf{compact and own upper semicontinuous} if, for all players $i \in I$, the strategy set $A_i$ is a compact topological space and the payoff function $u_i$ is upper semicontinuous in $a_i$ for any fixed $a_{-i} \in A_{-i}$. A game $\Gamma$ is \textbf{quasisupermodular} \citep[p. 102]{kultti-salonen1997undominated} if, for every $i\in I$, the strategy set $A_i$ is a complete lattice and the utility function $u_i$ is quasisupermodular and order upper semicontinuous in $a_i$ for any fixed $a_{-i} \in A_{-i}$, and $u_i$ satisfies the single crossing property in $(a_i, a_{-i})$.\footnote{We refer the reader to \cite{kultti-salonen1997undominated} for the formal definitions of quasisupermodularity, order upper semicontinuity, and the single crossing property.} The class of quasisupermodular games contains the class of \textit{games with strategic complementarities} introduced by \cite{milgrom-shannon1994monotone}, which in turn contains the class of \textit{supermodular games} of \cite{milgrom-roberts1990rationalizability}.

\begin{proposition}\label{prop:AllBounded}
All reductions of a game $\Gamma$ are completely bounded if any of the following is true:
\begin{itemize}
\item[1)] All strategy sets in $\Gamma$ are finite;
\item[2)] $\Gamma$ is a compact and own upper semicontinuous game;
\item[3)] $\Gamma$ is a quasisupermodular game.
\end{itemize}
\end{proposition}
\begin{proof}
See Appendix \ref{Proof:AllBounded}.
\end{proof}

In light of Corollary \ref{cor:AllBound}, in the games listed in Proposition \ref{prop:AllBounded} we can follow the nested elimination procedure without worrying. This also means that, instead of the commonly used universal procedure, we can use the nested procedure to find the Nash equilibria of a supermodular game, in which, as we know from \cite{milgrom-roberts1990rationalizability}, the join and meet of the set of serially undominated strategy profiles are the largest and smallest Nash equilibria of the game, respectively.

Corollary \ref{cor:AllBound} and part 2) of Proposition \ref{prop:AllBounded} extend Theorem 2 in \cite{chen2007iterated-strict-dominance} to the entire class of nested elimination sequences. Further, the equality $\mathcal{E}^n = \mathcal{E}^u$ in compact and own upper semicontinuous games is also proved by \cite{manili2026}, although his proof technique is different from ours.

\subsection{Examples}

\begin{example}\label{ex:UnboundedAtLimit}
Suppose the set of players is $I = \{1,2\}$, and strategy sets are $A_1 = [0,1]$ and $A_2 = \{\text{Left}, \text{Center}, \text{Right}\}$. The payoff function of player $1$ is
\begin{equation*}
u_1 (a_1, a_2) = \begin{dcases}
a_1 & \text{if} \; a_2 = \text{Left} \text{ or if } a_1 < 1 \text{ and } a_2 = \text{Right}\\
0 & \text{if} \; a_2 = \text{Center} \text{ or if } a_1 = 1 \text{ and } a_2 = \text{Right}.
\end{dcases}
\end{equation*}

The payoff function of player $2$ is
\begin{equation*}
u_2 (a_1, a_2) = \begin{dcases}
1 & \text{if} \; a_2 = \text{Left}\\
0 & \text{if} \; a_2 = \text{Center}\\
-1 & \text{if} \; a_2 = \text{Right}.
\end{dcases}
\end{equation*}

The unique maximal nested reduction of the game is $\hat{R}^n = \{1\}\times \left\lbrace \text{Left}\right\rbrace$, which is completely bounded. By Theorem \ref{thm:TFAE_outcomes}, $\{1\}\times \left\lbrace \text{Left}\right\rbrace$ is the maximal universal reduction too. Since there are elimination sequences that involve unbounded reductions, we have $\mathcal{E}^n \neq \mathcal{E}^u $ by Theorem \ref{thm:TFAE_classes}. To see this, consider the following universal elimination sequence:
\begin{equation}\label{eq:ExUnivNest}
A_1 \times A_2 \xlongrightarrow{u} A_1\times \left\lbrace \text{Left}, \text{Right} \right\rbrace \xlongrightarrow{u} \left([0, 0.5] \cup \{1\}\right)\times \left\lbrace \text{Left}, \text{Right} \right\rbrace  \xlongrightarrow{u} \{1\}\times \left\lbrace \text{Left}\right\rbrace. 
\end{equation}

The reduction $\{1\}\times \left\lbrace \text{Left} \right\rbrace$ is not a nested reduction of $\left([0, 0.5] \cup \{1\}\right)\times \left\lbrace \text{Left}, \text{Right} \right\rbrace$. Therefore, the sequence \eqref{eq:ExUnivNest} does not belong to $\mathcal{E}^n$.
\end{example}

\begin{example}[When order independence is not enough]\label{ex:OInd-no-equal}
Here we produce a game in which the class of nested elimination sequences is order independent yet $\mathbf{\hat{R}}^n \neq \mathbf{\hat{R}}^u$, from which it follows that $\mathcal{E}^n$ and $\mathcal{E}^u$ are disjoint. Suppose the set of players is $I = \{1,2\}$ and strategy sets are
\begin{align*}
	A_1 &= \left\lbrace \frac{2k}{2k+1} : k \geq 0 \right\rbrace \cup \{1\} =  \left\lbrace 0, \frac{2}{3}, \frac{4}{5},  \dots\right\rbrace \cup \{1\}\\
	A_2 &= \left\lbrace \frac{2k+1}{2k+2} : k \geq 0 \right\rbrace \cup \{1\} = \left\lbrace \frac{1}{2}, \frac{3}{4}, \frac{5}{6}, \dots\right\rbrace \cup \{1\}.
\end{align*}

The payoff function of each $i\in I$ is
\begin{equation*}
u_i (a_i, a_{-i}) = \begin{dcases}
0 & \text{if} \; a_i = a_{-i} = 1\\
\min\{a_i, a_{-i}\} & \text{otherwise}.
\end{dcases}
\end{equation*}

There is only one strictly dominated strategy across the two players, which is $0$. Once the latter is removed, the only dominated strategy is $\frac{1}{2}$. After removing $\frac{1}{2}$, the strategy $\frac{2}{3}$ becomes the only one to be dominated. Without $\frac{2}{3}$, the only dominated strategy is $\frac{3}{4}$, and so on. Formally, we can construct a sequence of reductions by letting $R^0 = A_1 \times A_2$ and, for all $\alpha$ such that $0 < \alpha < \omega$,
\begin{equation*}
	R^{\alpha} = R_1^{\alpha}\times R_2^{\alpha} = \left(A_1\setminus\left\lbrace \frac{2k}{2k + 1}: 0\leq k \leq \frac{\alpha-1}{2} \right\rbrace \right) \times \left(A_2\setminus\left\lbrace \frac{2k + 1}{2k + 2}: 0\leq k \leq \frac{\alpha-2}{2} \right\rbrace \right).
\end{equation*}

One can check that $R^{\alpha} \xlongrightarrow{n}R^{\alpha + 1}$ for all $\alpha < \omega$. At the first infinite ordinal, we have $R^{\omega} = \cap_{\alpha < \omega} R^{\alpha} = \{1\} \times \{1\}$, which is a maximal nested reduction. Hence, $\langle R^{\alpha}: \alpha \leq \omega \rangle$ is a feasible nested elimination sequence. Furthermore, $\{1\} \times \{1\}$ is the unique maximal nested reduction of the game, i.e., the class of nested elimination sequences is order independent. This follows from the fact that there is exactly one dominated strategy across players in any non-maximal nested reduction.

Finally, every non-maximal nested reduction $R^{\alpha}$ is completely bounded but $\hat{R}^n = \{1\} \times \{1\}$ is not. We know from Theorem \ref{thm:TFAE_outcomes} that $\hat{R}^n$ cannot be the maximal universal reduction. In fact, we have $\hat{R}^n \xlongrightarrow{u} \emptyset$ and, as a consequence, the maximal universal reduction is $\hat{R}^u = \emptyset$.
\end{example}

\begin{example}\label{ex:NotAllBounded}
We know from Corollary \ref{cor:AllBound} that the complete boundedness of all reductions of a game is sufficient for having $\mathcal{E}^n = \mathcal{E}^u$. But this example shows that it is not a necessary condition. Suppose the set of players is $I = \{1,2\}$, and strategy sets are $A_1 = [0,1]$ and $A_2 = \{\text{Left}, \text{Right}\}$. The payoff function of player $1$ is
	\begin{equation*}
		u_1 (a_1, a_2) = \begin{dcases}
			a_1 & \text{if} \; a_1 < 1 \text{ and } a_2 = \text{Left}\\
			0 & \text{otherwise}.
		\end{dcases}
	\end{equation*}
	
The payoff function of player $2$ is constant over $A = A_1 \times A_2$.

As there are no dominated strategies in this game, the classes $\mathcal{E}^n$ and $\mathcal{E}^u$ are completely bounded and the entire set of strategy profiles $A$ is the unique maximal reduction. Nevertheless, some reductions of this game are not completely bounded. A case in point is $[0,1] \times \{\text{Left}\}$. 
\end{example}

\section{An alternative formulation of nested reductions}\label{sec:GKZ}
Our definition of nested reductions follows \cite{dufwenberg2002existence-maximal-reduction} and several game theory textbooks. An alternative definition is provided in \cite{gilboa1990order}, which is one of the earliest articles on the order of eliminating dominated strategies.\footnote{More recently, \cite{hsieh-et-al_iterated-bounded-dominance2023} and \cite{qian-pathologies-IESD-2025} introduce another elimination procedure, called \textit{iterated elimination of boundedly strictly dominated strategies (IEBSDS)}, which delivers outcomes at least as coarse as those of the nested procedure. We refer the reader to their work for a detailed comparison between IEBSDS and the elimination procedures studied in this section.} In this section we compare the nested elimination sequences studied so far with those based on the alternative nested reductions of \cite{gilboa1990order}.

Given two reductions $R$ and $S$, we write $R \xlongrightarrow{\tilde{n}} S$ and say that $S$ is a \textbf{Gilboa-Kalai-Zemel (GKZ) reduction} of $R$ if the following holds: $S \subseteq R$ and, for all $i\in I$,
\begin{equation}\label{eq:AltNested-Red}
\text{if } a_i \in R_i\setminus S_i, \text{ then there is } b_i \in S_i \text{ such that } a_i \prec_{R_{-i}} b_i.
\end{equation}

The dominating strategy $b_i$ in \eqref{eq:AltNested-Red} must lie in $S_i \subseteq R_i$. Recall that, by contrast, the dominating strategy in a nested reduction is required to be in $R_i$. Clearly, every GKZ reduction is also nested but not vice versa. A GKZ elimination sequence is defined by letting $\ell = \tilde{n}$ in the definition of elimination sequences given in Subsection \ref{subsec:ElimSequen}. Note that every GKZ elimination sequence is the initial segment of a nested elimination sequence. The set of maximal GKZ reductions is $\mathbf{\hat{R}}^{\tilde{n}}$ and the class of all feasible GKZ sequences is $\mathcal{E}^{\tilde{n}}$.

As it turns out, we need the following, weaker form of boundedness to characterize the full equivalence between GKZ and nested elimination sequences.

\begin{definition}\label{Def:LBound}
A reduction $R$ is \textbf{locally bounded} if for all $i\in I$ and all $a_i \in R_i$ the following holds: if $D_{R_{-i}}(a_i) \cap R_i$ is non-empty, then there exists a strategy $c_i^* \in D_{R_{-i}}(a_i) \cap R_i$ that is not dominated relative to $R_{-i}$ by any $c_i \in R_i$. That is, if $a_i \prec_{R_{-i}} b_i$ for some $b_i \in R_i$, then there exists a strategy $c_i^* \in R_i$ such that $a_i \prec_{R_{-i}} c_i^*$ and, for all $c_i\in R_i$, it is not the case that $c_i^* \prec_{R_{-i}} c_i$.

A class of elimination sequences $\mathcal{E}^{\ell}$ is locally bounded if every reduction in the range of $\mathcal{E}^{\ell}$ is locally bounded.
\end{definition}

Notice that, in the definition above, the dominating strategies $b_i$ and $c_i^*$ must belong to $R_i \subseteq A_i$, and $c_i^*$ cannot be dominated by any strategy in $R_i$. Conversely, in a completely bounded reduction (Definition \ref{Def:CBound}), both $b_i$ and $c_i^*$ are required to be merely in $A_i$, and $c_i^*$ cannot be dominated by any strategy in $A_i$. Hence, any completely bounded reduction in the range of some elimination sequence\footnote{Outside the range of any feasible elimination sequence, complete boundedness does not necessarily imply local boundedness. For example, consider a one-player game in which the strategy space is $[0,1]$ and the payoff function is the identity map. Every reduction of this game is completely bounded, but the reduction $[0,1)$ is not locally bounded. Nevertheless, $[0,1)$ is not in the range of any feasible (nested or universal) elimination sequence.} is locally bounded but not vice versa. For instance, in Example \ref{ex:OInd-no-equal}, the maximal nested reduction $\hat{R}^n = \{1\} \times \{1\}$ is locally, but not completely, bounded.

As we discuss in more detail in Section \ref{subsec:RelConditions}, local boundedness generalizes the notion of \textit{games closed under dominance} introduced by \citet[p. 2016]{dufwenberg2002existence-maximal-reduction}.

\begin{theorem}\label{thm:TFAE_GKZ}
$\mathcal{E}^{\tilde{n}} \subseteq \mathcal{E}^{n}$ and $\mathbf{\hat{R}}^{\tilde{n}} = \mathbf{\hat{R}}^{n}$. Furthermore, the following statements are equivalent:
\begin{enumerate}
\item[1)] The class of GKZ elimination sequences $\mathcal{E}^{\tilde{n}}$ is locally bounded;
\item[2)] The class of nested elimination sequences $\mathcal{E}^n$ is locally bounded;
\item[3)] $\mathcal{E}^{\tilde{n}} = \mathcal{E}^n$.
\end{enumerate}
\end{theorem}
\begin{proof}
See Appendix \ref{Proof:GKZ}.
\end{proof}

Note that both $\mathcal{E}^{\tilde{n}} \subseteq \mathcal{E}^{n}$ and $\mathbf{\hat{R}}^{\tilde{n}} = \mathbf{\hat{R}}^{n}$ hold true without assuming local boundedness. In contrast to complete boundedness, local boundedness does not entail order independence: a case in point is Example \ref{ex:Closure*} in the next section.

We make use of the following result in the proof of Theorem \ref{thm:TFAE_GKZ}.
\begin{lemma}\label{lemma:GKZ-nested}
If $R \xlongrightarrow{n} S$, then there is a sequence of reductions $\langle T^{\alpha}: \alpha \leq \delta \rangle$, with $0 < \delta$, such that:
\begin{itemize}
\item $T^{0} = R$
\item $T^{\alpha} \xlongrightarrow{\tilde{n}} T^{\alpha + 1}$ for all $\alpha <\delta$
\item $T^{\lambda} = \cap_{\alpha < \lambda} T^{\alpha}$ for all limit ordinals $\lambda < \delta$
\item $T^{\delta} = S$.
\end{itemize}
\end{lemma}
\begin{proof}
See Appendix \ref{Proof:LemmaGKZ}.
\end{proof}

Lemma \ref{lemma:GKZ-nested} says that if $S$ is a nested reduction of $R$, then there is a non-increasing sequence of GKZ reductions that starts at $R$ and ends with $S$. As a result, given a nested elimination sequence, one can construct a GKZ elimination sequence having the same maximal reduction as the given sequence, so leading to the equality $\mathbf{\hat{R}}^{\tilde{n}} = \mathbf{\hat{R}}^{n}$ in Theorem \ref{thm:TFAE_GKZ}.

Lemma \ref{lemma:GKZ-nested} and the equality $\mathbf{\hat{R}}^{\tilde{n}} = \mathbf{\hat{R}}^{n}$ do not necessarily hold if one rules out GKZ sequences of transfinite length. To illustrate this point, consider again the one-player game having $(0,1)$ as its strategy set and the identity map as its payoff function. The empty set can be attained as a maximal nested reduction in just one step---that is, through the elimination sequence $(0,1) \xlongrightarrow{n} \emptyset$. But it is impossible to attain the empty set as a maximal GKZ reduction in finitely many steps. In fact, the fastest GKZ elimination sequences having the empty set as their maximal reduction are all indexed by the infinite ordinal $\omega + 1$. One such sequence is $\langle R^{\alpha}: \alpha < \omega +1\rangle$, where $R^{\alpha} = \left(\frac{\alpha}{\alpha + 1}, \frac{\alpha +1}{\alpha + 2}\right]$ for $ \alpha < \omega$ and $R^{\omega} = \cap_{\alpha <\omega} R^{\alpha}$.

Finally, recalling that a class of elimination sequences is completely bounded only if it is locally bounded, the following corollary is an immediate consequence of Theorems \ref{thm:TFAE_classes} and \ref{thm:TFAE_GKZ}.

\begin{corollary}
The following statements are equivalent:
\begin{enumerate}
\item[1)] The class of universal elimination sequences $\mathcal{E}^u$ is completely bounded;
\item[2)] The class of nested elimination sequences $\mathcal{E}^n$ is completely bounded;
\item[3)] The class of GKZ elimination sequences $\mathcal{E}^{\tilde{n}}$ is completely bounded;
\item[4)] $\mathcal{E}^u = \mathcal{E}^n = \mathcal{E}^{\tilde{n}}$.
\end{enumerate}
\end{corollary}

\section{Discussion}\label{sec:Discuss}

\subsection{Epistemic foundations}\label{subsec:Epistem}
Epistemic foundations for the universal elimination procedure are provided in \citet[Theorem 4]{chen2007iterated-strict-dominance}. They show that a strategy profile is compatible with common knowledge of rationality \textit{if and only if} it belongs to the maximal universal reduction. More specifically, they show the existence of a natural epistemic model for any game $\Gamma$, in which a strategy profile $a \in A$ is in $\hat{R}^u$ if and only if there is an epistemic state wherein $a$ is chosen and it is common knowledge that everyone is rational. A player is rational in a certain epistemic state if she does not play a strategy that is strictly dominated relative to the strategies she thinks the other players could play in that state. The notion of rationality adopted by \cite{chen2007iterated-strict-dominance} is based on \cite{aumann1995backward}, is formulated in terms of payoff dominance, and is a weaker version of the standard notion of rationality as subjective expected utility maximization. Using a similar notion of rationality, \citet[Proposition 1]{bach-cabessa2012common} carry out an epistemic analysis of the nested elimination sequence in which, at each stage, \textit{all} strictly dominated strategies are removed. In essence, they show that a strategy profile is compatible with common knowledge of rationality \textit{only if} it belongs to the maximal reduction of the nested sequence considered in their paper.

\subsection{Related conditions from the previous literature}\label{subsec:RelConditions}
Several conditions have been proposed that guarantee either the order independence of certain iterated elimination procedures or the equivalence between classes of elimination sequences. Here we discuss the conditions that are closest to our work. For ease of comparison, we reformulate these conditions using our notation and terminology.

\paragraph*{Property $\mathbf{C}$.} \cite{apt2007many-faces} introduces property $\mathbf{C}$ and then proves its sufficiency for the equivalence of the nested and the universal elimination sequences in which, at each stage, all strictly dominated strategies are removed.

\begin{definition}[\cite{apt2007many-faces}] The game $\Gamma$ satisfies property $\mathbf{C}$ if the following holds: for all $i\in I$ and all $a_i \in A_i$, if there are a strategy $b_i \in A_i$ and a reduction $R$ in the range of $\mathcal{E}^n$ such that $a_i \prec_{R_{-i}} b_i$, then there is a strategy $c^*_i \in A_i$ such that $a_i \prec_{R_{-i}} c^*_i$ and, for all $c_i \in A_i$, it is not the case that $c^*_i \prec_{R_{-i}} c_i$.
\end{definition}

It is easy to check that $\mathbf{C}$ implies the complete boundedness of $\mathcal{E}^n$, which in turn implies $\mathcal{E}^n = \mathcal{E}^u$ via Theorem \ref{thm:TFAE_classes}. Nevertheless, the following example shows that property $\mathbf{C}$ is not necessary for the equivalence of $\mathcal{E}^n$ and $\mathcal{E}^u$.

\begin{example}\label{ex:Apt}
The set of players is $I = \{1,2\}$ and strategy sets are
	\begin{align*}
		A_1 &= \{\text{Left}, \text{Right}\} \cup \left\lbrace \frac{2k}{2k+1} : k \geq 0 \right\rbrace =  \{\text{Left}, \text{Right}\} \cup \left\lbrace 0, \frac{2}{3}, \frac{4}{5},  \dots\right\rbrace \\
		A_2 &= \{\text{Left}, \text{Right}\} \cup \left\lbrace \frac{2k+1}{2k+2} : k \geq 0 \right\rbrace = \{\text{Left}, \text{Right}\} \cup \left\lbrace \frac{1}{2}, \frac{3}{4}, \frac{5}{6}, \dots\right\rbrace.
	\end{align*}
	The payoff function of each $i\in I$ is
	\begin{equation*}
		u_i (a_i, a_{-i}) = \begin{dcases}
			\min\{a_i, a_{-i}\} & \text{if} \; a_i, a_{-i} \in \mathbb{Q}\\
			a_i & \text{if} \; a_i \in \mathbb{Q} \text{ and } a_{-i} \in \{\text{Left}, \text{Right}\}\\
			1 & \text{if} \; a_i = a_{-i} = \text{Left} \text{ or } a_i = a_{-i} = \text{Right}\\
			0 & \text{otherwise}.
		\end{dcases}
	\end{equation*}
	
This is a game in which the two classes of elimination sequences $\mathcal{E}^n $ and $\mathcal{E}^u$ are both completely bounded and thus coincide. The unique maximal reduction is $\hat{R} = \{\text{Left}, \text{Right}\} \times \{\text{Left}, \text{Right}\}$, which is both nested and universal. But property $\mathbf{C}$ does not hold, and in particular it fails at $\hat{R}$ for any strategy that is a rational number.
\end{example}

\paragraph*{Forgetfulness-proof classes of elimination sequences.} \cite{hsieh-et-al_iterated-bounded-dominance2023} introduce a condition, called forgetfulness-proofness, which fully characterizes the equivalence between nested and GKZ elimination sequences.

\begin{definition}[\cite{hsieh-et-al_iterated-bounded-dominance2023}]\label{Def:Forget-proof}
The class of nested elimination sequences $\mathcal{E}^{n}$ is \textbf{forgetfulness-proof} if for every reduction $R$ in the range of $\mathcal{E}^{n}$ the following holds: for every $i\in I$, if $R_i\times R_{-i} \xlongrightarrow{n} S_i\times R_{-i}$ for some $S_i\subseteq R_i$ and $a_i\in S_i$ is dominated relative to $R_{-i}$ by some $b_i \in R_i$, then $a_i$ is dominated relative to $R_{-i}$ by some $c_i \in S_i$.
\end{definition}

One can show that $\mathcal{E}^{n}$ is forgetfulness-proof if and only if it is locally bounded. But the same does not hold for the class of GKZ elimination sequences: if one adapts Definition \ref{Def:Forget-proof} to $\mathcal{E}^{\tilde{n}}$ by substituting $\xlongrightarrow{\tilde{n}}$ for $\xlongrightarrow{n}$, then it is possible to show that $\mathcal{E}^{\tilde{n}}$ is always forgetfulness-proof but may fail to be locally bounded.

\paragraph*{Games closed under dominance*.} Local boundedness (Definition \ref{Def:LBound}) generalizes the notion of \textit{games closed under dominance} introduced by \citet[p. 2016]{dufwenberg2002existence-maximal-reduction}. While closure under dominance is defined for reductions attainable in finitely many elimination steps, local boundedness of a class $\mathcal{E}^{\ell}$ is defined over the full range of the class. A different generalization, called \textit{closure under dominance*}, is offered by \cite{luo2020iterated}. 

\begin{definition}[\cite{luo2020iterated}] The game $\Gamma$ is \textbf{closed under dominance*} if, for all reductions $R$ and $S$ for which there is a nested elimination sequence $\langle R^{\alpha}: \alpha < \delta \rangle$ wherein $R^{\beta} = R$ and $R^{\gamma} = S$ for some $\beta \leq \gamma < \delta$, the following is true: if $a_i \prec_{R_{-i}} b_i$ for some $i\in I$, $a_i \in S_i \subseteq R_i$ and $b_i\in R_i$, then there exists a strategy $c_i^* \in S_i$ such that $a_i \prec_{S_{-i}} c^*_i$ and, for all $c_i \in S_i$, it is not the case that $c^*_i \prec_{S_{-i}} c_i$.
\end{definition}

\cite{luo2020iterated} prove that closure under dominance* is sufficient for the order independence of the nested elimination procedure and for its equivalence to the elimination procedure of \cite{gilboa1990order}. Nevertheless, closure under dominance* is not sufficient for the equivalence of the nested and universal elimination procedures. To see this, notice that the game in Example \ref{ex:OInd-no-equal} is closed under dominance*, and yet its classes $\mathcal{E}^n$ and $\mathcal{E}^u$ are different in that they both fail to be completely bounded. On the other hand, closure under dominance* implies local boundedness of $\mathcal{E}^n$. The converse does not hold though, as the following example shows. 

\begin{example}\label{ex:Closure*}
The set of players is $I = \{1,2\}$ and strategy sets are
	\begin{align*}
		A_1 &= \{-1\} \cup \left\lbrace \frac{2k}{2k+1} : k \geq 0 \right\rbrace =  \left\lbrace -1, 0, \frac{2}{3}, \frac{4}{5},  \dots\right\rbrace \\
		A_2 &= \left\lbrace \frac{2k+1}{2k+2} : k \geq 0 \right\rbrace \cup \{ 1 \} = \left\lbrace \frac{1}{2}, \frac{3}{4}, \frac{5}{6}, \dots\right\rbrace \cup \{ 1 \}.
	\end{align*}
	The payoff function of player $1$ is
	\begin{equation*}
		u_1 (a_1, a_{2}) = \begin{dcases}
			- 1 & \text{if} \; a_1 =  -1 \\
			\min\{a_1, a_{2}\} & \text{otherwise}
		\end{dcases}
	\end{equation*}
while that of player $2$ is
\begin{equation*}
		u_2 (a_1, a_{2}) = \begin{dcases}
			a_2 & \text{if} \; a_1 =  -1 \\
			\min\{a_1, a_{2}\} & \text{otherwise}.
		\end{dcases}
	\end{equation*}

The class of nested elimination sequences of this games is locally bounded. Nevertheless, the game is not closed under dominance*. To see why, consider the maximal nested reduction $\hat{R}^n = \{-1\} \times \{1\}$. Since the strategy $a_1 = -1$ is strictly dominated at the very beginning of every elimination sequence, closure under dominance* says that Player $1$ must have a strategy in $\hat{R}^n = \{-1\} \times \{1\}$ that dominates $a_1 = -1$, which is clearly not the case. To conclude, notice that the empty reduction is another maximal nested reduction of this game, so proving that $\mathcal{E}^n$ is order dependent too.
\end{example}

\newpage
\appendix
\section{Appendix}

\subsection{Proof of Proposition \ref{prop:AllBounded}}\label{Proof:AllBounded}
\begin{itemize}
\item[1)] If all strategy sets in $\Gamma$ are finite, it is clear that every reduction of $\Gamma$ is completely bounded.
\item[2)] The proof of this part is adapted from the proof of the Lemma in \citet[p. 2012]{dufwenberg2002existence-maximal-reduction}. Suppose $\Gamma$ is a compact and own upper semicontinuous game. Let $R$ be a reduction of $\Gamma$ and suppose there are a player $i\in I$ and strategies $a_i \in R_i$ and $b_i \in A_i$ such that $a_i \prec_{R_{-i}} b_i$. We need to show that there is a strategy $c_i^* \in A_i$ such that $a_i \prec_{R_{-i}} c_i^*$ and, for all $c_i \in A_i$, it is not the case that $c_i^* \prec_{R_{-i}} c_i$.

Consider the following set
\begin{equation*}
Z:=	\bigcap_{a_{-i} \in R_{-i}} \left\lbrace c_i \in A_i: u_i(c_i, a_{-i}) \geq u_i(b_i, a_{-i}) \right\rbrace.
\end{equation*}
Clearly, $Z$ is non-empty as it contains $b_i$. For any $a_{-i} \in R_{-i}$, the set 
\begin{equation*}
\left\lbrace c_i \in A_i: u_i(c_i, a_{-i}) \geq u_i(b_i, a_{-i}) \right\rbrace
\end{equation*}
is closed due to the upper semicontinuity of $u_i$ in $c_i$. Hence, $Z$ is closed as it is the intersection of closed sets. Being a closed subset of the compact space $A_i$, the set $Z$ is compact as well.

Now fix an arbitrary $a^*_{-i} \in R_{-i}$ and let $u_{i\vert Z}: Z \to \mathbb{R}$ be the restriction of $u_i(c_i, a^*_{-i})$ to $Z$. Since $u_i(c_i, a^*_{-i})$ is upper semicontinuous in $c_i$, its restriction $u_{i\vert Z}$ is upper semicontinuous too. And since $Z$ is compact, $u_{i\vert Z}$ attains a maximum value at some $\hat{c}_i \in Z$. Since $\hat{c}_i \in Z$ and $a_i \prec_{R_{-i}} b_i$, we have $a_i \prec_{R_{-i}} \hat{c}_i$. Furthermore, since $\hat{c}_i$ maximizes $u_{i\vert Z}$, there is no $c_i \in A_i$ that strictly dominates $\hat{c}_i$ relative to $R_{-i}$. Therefore, we can choose $c_i^* = \hat{c}_i$ to prove the claim.

\item[3)] Suppose $\Gamma$ is a quasisupermodular game as per \citet[p. 102]{kultti-salonen1997undominated}. By Theorem 20 in \citet[p. 250]{birkhoff1967lattice3rd}, any complete lattice is compact in the order interval topology. By Theorem A3 in \citet[p. 179]{milgrom-shannon1994monotone}, an order upper semicontinuous and quasisupermodular function is upper semicontinuous in the order interval topology. Therefore, we can apply part 2) to conclude that every reduction of $\Gamma$ is completely bounded. \qed
\end{itemize}

\subsection{Proof of Theorem \ref{thm:TFAE_GKZ}}\label{Proof:GKZ}
\begin{itemize}
\item Here we show $\mathcal{E}^{\tilde{n}} \subseteq \mathcal{E}^{n}$. Take any GKZ elimination sequence and let $\hat{R}^{\tilde{n}}$ be its maximal reduction. Since any GKZ sequence is the initial segment of a nested sequence, all we have to prove is that $\hat{R}^{\tilde{n}}$ is a maximal nested reduction. By way of contradiction, suppose there is a reduction $S$ such that $S \neq \hat{R}^{\tilde{n}}$ and $\hat{R}^{\tilde{n}} \xlongrightarrow{n} S$. Then there are a player $i\in I$, a strategy $a_i \in \hat{R}^{\tilde{n}}_i \setminus S_i$ and a strategy $b_i \in \hat{R}^{\tilde{n}}_i$ such that $a_i \prec_{\hat{R}^{\tilde{n}}_{-i}} b_i$. But then we have $\hat{R}^{\tilde{n}} \xlongrightarrow{\tilde{n}} \left(\hat{R}^{\tilde{n}}_i \setminus \{a_i\}\right)\times \hat{R}^{\tilde{n}}_{-i}$, which contradicts the maximality of the GKZ reduction $\hat{R}^{\tilde{n}}$.

\item The inclusion $\mathbf{\hat{R}}^{\tilde{n}} \subseteq \mathbf{\hat{R}}^{n}$ follows from $\mathcal{E}^{\tilde{n}} \subseteq \mathcal{E}^{n}$. To show the reverse inclusion, take any nested elimination sequence $\langle R^{\alpha}: \alpha < \delta \rangle$. For every element $R^{\alpha}$ in this sequence, if it is not the case that $R^{\alpha} \xlongrightarrow{\tilde{n}} R^{\alpha + 1}$, then construct a sequence from $R^{\alpha}$ to $R^{\alpha + 1}$ as per Lemma \ref{lemma:GKZ-nested}. Denote the latter sequence by $\langle R^{\alpha}, R^{\alpha + 1}\rangle$. Now define the set $\mathbf{T}$ as the union of $\left\lbrace R^{\alpha}: \alpha < \delta \right\rbrace$ and the range of all the sequences $\langle R^{\alpha}, R^{\alpha + 1}\rangle$ formed previously. Now, since $\langle R^{\alpha}: \alpha < \delta \rangle$ and every sequence $\langle R^{\alpha}, R^{\alpha + 1}\rangle$ are non-increasing, the set $\mathbf{T}$ is well-ordered by reverse inclusion. As such, $\mathbf{T}$ is isomorphic to a unique ordinal $\delta'$ \citep[p. 36]{roman2008lattices}. The unique isomorphism from $\delta'$ to $\mathbf{T}$ is a well-defined GKZ sequence, and it has the same maximal reduction as $\langle R^{\alpha}: \alpha < \delta \rangle$. This proves that for any nested sequence, we can find a GKZ sequence with the same maximal reduction, whence $\mathbf{\hat{R}}^{n} \subseteq \mathbf{\hat{R}}^{\tilde{n}}$.

\item $\bm{1) \implies 3).}$  Having already established $\mathcal{E}^{\tilde{n}} \subseteq \mathcal{E}^{n}$, it is enough to show by induction that, if $\mathcal{E}^{\tilde{n}}$ is locally bounded, every initial segment of a nested elimination sequence is the initial segment of a GKZ sequence, from which it follows that every nested sequence is in $\mathcal{E}^{\tilde{n}}$. The proof can be adapted easily from the proof of Theorem \ref{thm:TFAE_classes}.
\item $\bm{3) \implies 2).}$ Suppose by way of contradiction that $\mathcal{E}^{\tilde{n}} = \mathcal{E}^{n}$ but a reduction $R$ in the range of $\mathcal{E}^{n}$ is not locally bounded. That is, there are a player $i \in I$ and a strategy $a_i \in R_i$ such that $D_{R_{-i}}(a_i) \cap R_i \neq \emptyset$ and every strategy in $D_{R_{-i}}(a_i) \cap R_i$ is dominated by another strategy in the same set. This would lead to a contradiction since
\begin{equation*}
R \xlongrightarrow{n} \left(R_i \setminus \left((D_{R_{-i}}(a_i) \cap R_i)\cup \{a_i\}\right)\right) \times R_{-i},
\end{equation*}
but $\left(R_i \setminus \left((D_{R_{-i}}(a_i) \cap R_i)\cup \{a_i\}\right)\right) \times R_{-i}$ is not a GKZ reduction of $R$.
\item $\bm{2) \implies 1).}$ Since every GKZ sequence is the initial segment of a nested sequence, the range of $\mathcal{E}^{\tilde{n}}$ is contained in the range of $\mathcal{E}^{n}$, from which the claim follows. \qed
\end{itemize}

\subsection{Proof of Lemma \ref{lemma:GKZ-nested}}\label{Proof:LemmaGKZ}
The proof is arranged in three parts.

\paragraph{Part 1)} Let $R \xlongrightarrow{n} S$. The statement is trivial if $R = \emptyset$ or $R=S$. Suppose $R \neq \emptyset$ and $R\neq S$. For any $i\in I$ and $a_i\in A_i$, define the \textbf{strict lower contour set} of $a_i$ relative to $R_{-i}$ as the set $L_{R_{-i}}(a_i) := \left\lbrace b_i \in A_i: b_i \prec_{R_{-i}} a_i\right\rbrace$. Construct a sequence of reductions $\langle T^{\alpha}\rangle$, where $\alpha$ ranges over all ordinals and $T^{\alpha}= \prod_{i\in I}T_i^{\alpha}$ for all $\alpha$, as follows: For all $i\in I$,
\begin{itemize}
\item $T_i^0 = S_i \cup Z_i^0 \cup Y_i$, where
\begin{align*}
Z_i^0 &= \left\lbrace a_i \in R_i\setminus S_i : \not\exists b_i \in S_i \text{ such that } a_i \prec_{R_{-i}}b_i \right\rbrace\\
Y_i &= \left\lbrace a_i \in R_i\setminus S_i : \exists b_i \in S_i \text{ such that } a_i \prec_{R_{-i}}b_i \right\rbrace
\end{align*}
\item $T_i^{\alpha + 1} = S_i \cup Z_i^{\alpha + 1}$ for all ordinals $\alpha$, where
\begin{equation*}
		Z_i^{\alpha + 1} = \begin{dcases}
			Z_i^{\alpha} \setminus \left(L_{R_{-i}}(a_i) \cap Z_i^{\alpha}\right) \text{ for some } a_i \in Z_i^{\alpha} \text{ s.t. } L_{R_{-i}}(a_i) \cap Z_i^{\alpha} \neq \emptyset & \text{if} \; Z_i^{\alpha} \neq \emptyset\\
			Z_i^{\alpha} & \text{otherwise}
		\end{dcases}
	\end{equation*}
\item $T_i^{\lambda} = \cap_{\alpha < \lambda} T_i^{\alpha}$ for all limit ordinals $\lambda$.
\end{itemize}

It is easy to check that $S_i$, $Z_i^0$ and $Y_i$ are pairwise disjoint and $T_i^0 = R_i$. Furthermore, the sequence is non-increasing, i.e., $T^{\beta} \subseteq T^{\alpha}$ whenever $\alpha < \beta$.

\paragraph{Part 2)} Here we show by induction that, for all ordinals $\alpha$, if $a_i \in Z_i^{\alpha}$, then there exists a strategy $b_i \in Z_i^{\alpha}$ such that $a_i \in L_{R_{-i}}(b_i)$. Suppose $\alpha = 0$ and $a_i \in Z_i^{0}$. Then it follows from the definition of $Z_i^0$ and $Y_i$ that there is $b_i \in Z_i^{0}$ such that $a_i \in L_{R_{-i}}(b_i)$. For the inductive step, suppose for all $\beta < \alpha$, if $a_i \in Z_i^{\beta}$, then there exists $b_i \in Z_i^{\beta}$ such that $a_i \in L_{R_{-i}}(b_i)$. Let $a_i \in Z_i^{\alpha}$ and take any $\beta' < \alpha$. Since $Z_i^{\alpha} \subseteq Z_i^{\beta'}$, by the inductive hypothesis there is a strategy $b_i \in Z_i^{\beta'}$ such that $a_i \in L_{R_{-i}}(b_i)$. Now, if $b_i \in Z_i^{\alpha}$, then the claim clearly holds. If $b_i \notin Z_i^{\alpha}$, then $b_i \in Z_i^{\gamma}\setminus Z_i^{\gamma + 1}$ for some $\gamma$ such that $\beta'\leq \gamma < \alpha$. Thus we have
\begin{equation*}
Z_i^{\gamma}\setminus Z_i^{\gamma + 1} = Z_i^{\gamma} \cap L_{R_{-i}}(c_i) \text{ for some } c_i \in Z_i^{\gamma}.
\end{equation*}

Now, $a_i \in L_{R_{-i}}(c_i)$ by transitivity of the strict dominance relation, and $a_i \in Z_i^{\alpha} \subseteq Z_i^{\gamma}$ by assumption. At the same time, $a_i \notin Z_i^{\gamma + 1} \supseteq Z_i^{\alpha}$, so leading to a contradiction. Therefore, we must have $b_i \in Z_i^{\alpha}$.

\paragraph{Part 3)} The claim proved in part 2) entails that, for all $\alpha$, if $Z_i^{\alpha}$ is non-empty, so is $Z_i^{\alpha}\setminus Z_i^{\alpha + 1}$. Therefore, recalling that our sequence $\langle T^{\alpha}\rangle$ is non-increasing, there is a least ordinal $\delta$ at which $Z_i^{\delta} = \emptyset$ and $T_i^{\delta} = S_i$ for all $i\in I$.

It remains to show that $T^{\alpha} \xlongrightarrow{\tilde{n}} T^{\alpha + 1}$ for all $\alpha <\delta$. By construction, if $a_i \in Z_i^{\alpha}\setminus Z_i^{\alpha + 1}$, then there is a $b_i\in Z_i^{\alpha}$ such that $a_i \in L_{R_{-i}}(b_i)$. Hence, $b_i$ lies in $Z_i^{\alpha + 1}$ and strictly dominates $a_i$. This shows $T^{\alpha} \xlongrightarrow{\tilde{n}} T^{\alpha + 1}$ for all $\alpha <\delta$, so ending the proof. \qed 

\newpage

\bibliographystyle{plainnat}
\bibliography{biblio}

\end{document}